\documentclass[11pt,reqno]{amsart}
\usepackage{amscd,amssymb,amsmath,amsthm}
\usepackage[arrow,matrix]{xy}
\usepackage{graphicx}
\usepackage{cite}
\newtheorem{thm}[subsection]{Theorem}
\newtheorem{lemma}[subsection]{Lemma}
\newtheorem{pro}[subsection]{Proposition}

\newtheorem{rk}[subsection]{Remark}
\newtheorem{defn}[subsection]{Definition}

\numberwithin{equation}{section} 

\newcommand{\bea}{\begin{eqnarray}}
\newcommand{\eea}{\end{eqnarray}}

\newcommand{\Z}{\mathbb{Z}}
\newcommand{\Q}{\mathbb{Q}}

\begin{document}
\title[On $p$-adic Gibbs measures of HC model]{On $p$-adic Gibbs Measures for Hard Core Model on a Cayley Tree}

\author{D. Gandolfo, U. A. Rozikov, J. Ruiz}

 \address{D.\ Gandolfo and J.\ Ruiz\\Centre de Physique Th\'eorique, UMR 6207, Universit\'es Aix-Marseille
 et Sud Toulon-Var, Luminy Case 907, 13288 Marseille, France.}
\email {gandolfo@cpt.univ-mrs.fr\ \ ruiz@cpt.univ-mrs.fr}

 \address{U.\ A.\ Rozikov\\ Institute of mathematics and information technologies,
29, Do'rmon Yo'li str., 100125, Tashkent, Uzbekistan.}
\email {rozikovu@yandex.ru}

\begin{abstract} In this paper we consider a nearest-neighbor $p$-adic hard core (HC) model, with fugacity $\lambda$, on a homogeneous Cayley tree of order $k$ (with $k + 1$ neighbors). We focus on $p$-adic Gibbs measures for the HC model, in particular on $p$-adic "splitting" Gibbs measures generating a $p$-adic Markov chain along each path on the tree. We show that the $p$-adic HC model is completely different from real HC model: For a fixed $k$ we prove that the $p$-adic HC model may have a splitting Gibbs measure only if $p$ divides $2^k-1$. Moreover if $p$ divides $2^k-1$ but does not divide $k+2$ then there exists unique translational invariant $p$-adic Gibbs measure. We also study $p$-adic periodic splitting Gibbs measures
and show that the above model admits only translational invariant and periodic with
period two (chess-board) Gibbs measures. For $p\geq 7$ (resp. $p=2,3,5$) we give necessary and sufficient (resp. necessary) conditions for the existence of a periodic $p$-adic measure.   For $k=2$ a $p$-adic splitting Gibbs measures exists if and only if $p=3$,
in this case we show that if $\lambda$ belongs to a $p$-adic ball of radius $1/27$ then there are precisely two periodic (non translational invariant) $p$-adic Gibbs measures. We prove that a $p$-adic Gibbs measure is bounded if and only if $p\ne 3$.
\end{abstract}
\maketitle



\section{Introduction} \label{sec:intro}

In \cite{SR} a hard core (HC) model with nearest neighbor
interaction and spin values $0,1$ on a Cayley tree was studied.
In this paper we consider $p$-adic version of this model.

One of the central problems in the theory of Gibbs measures of
lattice systems is to describe infinite-volume (or limiting) Gibbs
measures corresponding to a given Hamiltonian. A complete analysis
of this set is often a difficult problem. Many papers have been
devoted to these studies when the underlying lattice is a Cayley
tree \cite{y1,y2,B,G, yp1, 17, MRS, y11, yp2,yp3,SR,y16}.

In all these works the models under consideration have a finite set of spin values on the field of real numbers. These models have the following common property:
the  existence  of finitely many translation-invariant  and uncountable
numbers  of  non-translation-invariant extreme Gibbs
measures.
Also for several models it was proved that there exist periodic Gibbs measures (which are invariant with respect to normal  subgroups  of  finite index of the group representation of Cayley tree) and that there are uncountable number of non-periodic Gibbs measures.

On the other hand, various models described in the language of $p$-adic analysis
have been actively studied, see e.g. \cite{5, 14, 34, 49} and numerous applications of $p$-adic analysis to mathematical physics have been proposed
in \cite{7, 23, 24, 25}. Well-known studies in this area were devoted
to quantum mechanical models \cite{50, 48}. One of the first applications of $p$-adic numbers in quantum physics appeared in the framework of quantum logic \cite{8}.
This model is of special interest to us because it cannot be
described by using conventional real-valued probability.

It is also known \cite{24, 29, 34, 48} that
a number of $p$-adic models in physics cannot be described using ordinary Kolmogorov's probability
theory. In \cite{28} an abstract $p$-adic probability theory was developed by means of the theory
of non-Archimedean measures \cite{40}.

A non-Archimedean analog of the Kolmogorov theorem was proved in
\cite{16}. Such a result allows to construct wide classes of
stochastic processes and the possibility to develop statistical
mechanics in the context of $p$-adic theory.

We refer the reader to \cite{19}, \cite{26},\cite{Fg},
\cite{36}-\cite{38} where  various models of statistical physics
in the context of $p$-adic field are studied.

In the present paper we consider $p$-adic Gibbs measures of a hard core model on the Cayley tree  over the $p$-adic field (we refer the reader to \cite{SR} for the real case).

 The paper is organized as follows. Section 2 presents definitions and known results.  Section 3 is devoted to the standard construction of ($p$-adic) Gibbs measures characterized by a functional equation.  Section 4 contains conditions of solvability of this equation.
Under conditions on $p$ and on the degree $k$ of the tree, we prove in Section 5 the existence and uniqueness of translational-invariant $p$-adic Gibbs measure. In Section 6 we study $p$-adic periodic Gibbs measures
and show that the HC model admits only translational invariant and periodic with
period two (chess-board) Gibbs measures. For $k=2$ a $p$-adic splitting Gibbs measures exists if and only if $p=3$,
in this case we show that if $\lambda$ belongs to a $p$-adic ball of radius $1/27$ then there are precisely two periodic (non translational invariant) $p$-adic Gibbs measures. In Section 7 we prove that a $p$-adic Gibbs measure is bounded if and only if $p\ne 3$. In the last section devoted to concluding remarks, we present comparisons between real and $p$-adic Gibbs measures.

\section{Preliminaries}

\subsection{$p$-adic numbers and measures.} Let $\Q$ be the field of rational numbers. For a fixed prime number $p$, every rational number $x\ne 0$ can be represented
in the form $x = p^r{n\over m}$, where $r, n\in \Z$, $m$ is a positive integer, and $n$ and $m$ are relatively prime with $p$: $(p, n) = 1$, $(p, m) = 1$. The $p$-adic norm of $x$ is given by
$$|x|_p=\left\{\begin{array}{ll}
p^{-r}\ \ \mbox{for} \ \ x\ne 0\\
0\ \ \mbox{for} \ \ x = 0.
\end{array}\right.
$$
This norm is non-Archimedean  and satisfies the so called strong triangle inequality
$$|x+y|_p\leq \max\{|x|_p,|y|_p\}.$$

We will often use the following fact:
If $|x|_p\ne |y|_p$ then
$$|x+y|_p=\max\{|x|_p,|y|_p\}.$$

The completion of $\Q$ with respect to the $p$-adic norm defines the $p$-adic field
 $\Q_p$. Any $p$-adic number $x\ne 0$ can be uniquely represented
in the canonical form
\begin{equation}\label{ek}
x = p^{\gamma(x)}(x_0+x_1p+x_2p^2+\dots),
\end{equation}
where $\gamma(x)\in \Z$ and the integers $x_j$ satisfy: $x_0 > 0$, $0\leq x_j \leq p - 1$ (see
\cite{29,41,48}). In this case $|x|_p = p^{-\gamma(x)}$.

\begin{thm}\label{tx2} \cite{29}, \cite{48} The equation
$x^2 = a$, $0\ne a =p^{\gamma(a)}(a_0 + a_1p + ...), 0\leq a_j \leq p - 1$, $a_0 > 0$
has a solution $x\in \Q_p$ if and only if the following conditions
are fulfilled:

i) $\gamma(a)$ is even;

ii) $a_0$ is a quadratic residue modulo $p$ if $p\ne 2$; $a_1 = a_2 = 0$ if $p = 2$.
\end{thm}

The elements of the set $\mathbb{Z}_p=\{x\in \Q_p: |x|_p\leq 1\}$ are called $p$-adic integers.

The following statement is known as Hensel's lemma \cite{29}.

\begin{thm}\label{hl} Let $F(x)=\sum_{i=0}^nc_ix^i$ be a polynomial whose coefficients are $p$-adic integers. Let
 $F'(x)=\sum_{i=0}^nic_ix^{i-1}$ be the derivative of $F(x)$. Let $a_0$ be a $p$-adic integer such that
 $F(a_0)\equiv 0 \,(\operatorname{mod} p)$ and  $F'(a_0)\neq 0 \,(\operatorname{mod} p)$. Then there exists a unique $p$-adic integer $a$ such that
 $F(a)=0$ and $a=a_0\,(\operatorname{mod} p)$.
 \end{thm}

Given $a\in \Q_p$ and $r > 0$ put
$$B(a, r) = \{x\in \Q_p : |x-a|_p < r\}.$$

The $p$-adic {\it logarithm} is defined by the series
$$\log_p(x) =\log_p(1 + (x-1)) =
\sum_{n=1}^\infty (-1)^{n+1}{(x-1)^n\over n},$$
which converges for $x\in B(1, 1)$; the $p$-adic exponential is defined by
$$\exp_p(x) =\sum^\infty_{n=0}{x^n\over n!},$$
which converges for $x \in B(0, p^{-1/(p-1)})$.

\begin{lemma}\label{l1} \cite{29, 41}. Let $x\in B(0, p^{-1/(p.1)}$, then
$$|\exp_p(x)|_p = 1,\ \ |\exp_p(x)-1|_p = |x|_p, \ \ |\log_p(1 + x)|_p = |x|_p,$$
$$\log_p(\exp_p(x)) = x,\ \ \exp_p(\log_p(1 + x)) = 1 + x.$$
\end{lemma}

We refer the reader to \cite{29, 41, 48} for the basics of $p$-adic analysis and $p$-adic mathematical physics.

Let $(X,{\mathcal B})$ be a measurable space, where ${\mathcal B}$ is an algebra of subsets of $X$. A function $\mu: {\mathcal B}\to \Q_p$
is said to be a $p$-adic measure if for any $A_1, . . . ,A_n\in {\mathcal B}$ such that $A_i\cap A_j = \emptyset$, $i\ne j$, the following holds:
$$\mu(\bigcup^n_{j=1}A_j)=\sum^n_{j=1}\mu(A_j).$$
A $p$-adic measure is called a probability measure if $\mu(X) = 1$, see, e.g. \cite{22}, \cite{40}.

\subsection{Cayley tree.} The Cayley tree (Bethe lattice \cite{y1}) $\Gamma^k$
of order $ k\geq 1 $ is an infinite tree, i.e., a graph without
cycles, such that exactly $k+1$ edges originate from each vertex.
Let $\Gamma^k=(V, L)$ where $V$ is the set of vertices and  $L$ the set of edges.
Two vertices $x$ and $y$ are called {\it nearest neighbors} if there exists an
edge $l \in L$ connecting them.
We will use the notation $l=\langle x,y\rangle$.
A collection of nearest neighbor pairs $\langle x,x_1\rangle, \langle x_1,x_2\rangle,...,\langle x_{d-1},y\rangle$ is called a {\it
path} from $x$ to $y$. The distance $d(x,y)$ on the Cayley tree is the number of edges of the shortest path from $x$ to $y$.

For a fixed $x^0\in V$, called the root, we set
\begin{equation*}
W_n=\{x\in V| d(x,x^0)=n\}, \qquad V_n=\bigcup_{k=1}^n W_k
\end{equation*}
and denote
$$
S(x)=\{y\in W_{n+1} :  d(x,y)=1 \}, \ \ x\in W_n, $$ the set  of {\it direct successors} of $x$.

\subsection{Hard Core model} We consider HC model with nearest neighbor interactions on a Cayley tree
where the spins assigned to the vertices of the tree take values in the set $\Phi:=\{0,1\}$.
A configuration $\sigma$ on $A\subset V$ is then defined as a function $x\in A\mapsto\sigma (x)\in\Phi$.
The set of all configurations is $\Phi^A$. A site $x$ is called ``occupied" if $\sigma(x)=1$ and ``vacant" if $\sigma(x)=0$.

A configuration is called {\it admissible} if the product $\sigma(x)\sigma(y)=0$ for all nearest neighbor pair $\langle x,y\rangle$.
We denote $\Omega_A$ the set of all admissible configurations on $A\subset V$ and set
 $\Omega=\Omega_V$.

Let $p$ be a fixed prime number.
The (formal) $p$-adic Hamiltonian of the HC model is the mapping $H:\Omega\to \Q_p$ given by
\begin{equation}\label{h}
H(\sigma)=-J\sum_{x\in V}
\sigma(x),
\end{equation}
where $J\in \Q_p$ is a constant such that
\begin{equation}\label{J}
|J|_p< p^{-1/(p-1)}.
\end{equation}

Note that such a condition provides the existence of a $p$-adic Gibbs measure defined through the $p$-adic exponential. As it was mentioned above the set of values of a $p$-adic norm $|\cdot|_p$ is $\{p^m: m\in \Z\}$, consequently
the condition (\ref{J}) is equivalent to the condition $|J|_p\leq {1\over p}$.

\section{Construction of $p$-adic Gibbs measure}
Let us construct $p$-adic Gibbs measures of this HC model. Since we use $\exp_p(x)$
to define the $p$-adic Gibbs measure, all quantities which arise below must belong to the set:

$$\mathcal E_p=\{x\in Q_p: |x|_p=1, |x-1|_p< p^{-1/(p-1)}\}.$$

As in classical (real) case we consider a special class
of Gibbs measures. We call them $p$-adic splitting Gibbs measures, a formal
definition follows.

Write $x < y$ if the pathes from $x^0$ to $y$ go through $x$. By this notation a vertex $y$ is a direct
successor of $x$ if $y > x$ and $x, y$ are nearest neighbors. Note that the root $x^0$ has $k + 1$ direct successors and any vertex $x \ne x^0$ has $k$ direct successors.

Let $z : x\to  z_x = (z_{0,x}, z_{1,x})\in {\mathcal E}^2_p$ be a vector-valued function on $V$, we will  consider the $p$-adic probability measures on $\Omega_{V_n}$ defined by
\begin{equation}\label{e2.1a}
\mu^{(n)}(\sigma_n)={1\over Z_n}\exp_p\left(J\sum_{x\in V_n}\sigma(x)\right)\prod_{x\in W_n}z_{\sigma(x),x}
\end{equation}
where $Z_n$ is the corresponding partition function:
\begin{equation}\label{e2.1b}
Z_n =\sum_{\varphi\in \Omega_{V_n}}\exp_p\left(J\sum_{x\in V_n}\varphi(x)\right)\prod_{x\in W_n}z_{\varphi(x),x}.
\end{equation}
Let us mention that function $z$ plays the role of a generalized boundary condition.

One of the central results of probability theory concerns the construction of an infinite-volume
distribution with given finite-dimensional distributions.
In this paper we consider this problem in $p$-adic context. More
precisely, we want to define a $p$-adic probability measure $\mu$ on the set $\Omega$ of admissible configurations. In general, the existence of such a measure is not known, since there is not enough
information on the topological properties of the set of all $p$-adic measures
defined even on compact spaces.
Therefore, we can only use the $p$-adic Kolmogorov extension theorem (see \cite{16, 33}) based on the so-called compatibility condition.

We say that the $p$-adic probability measures $\mu^{(n)}$ are compatible if for all $n\geq 1$ and $\sigma_{n-1}\in \Omega_{V_{n-1}}$:
\begin{equation}\label{e2.2}
\sum_{\omega_n\in \Omega_{W_n}}\mu^{(n)}(\sigma_{n-1}\vee\omega_n){\mathbf 1}(\sigma_{n-1}\vee\omega_n\in \Omega_{V_n})= \mu^{(n-1)}(\sigma_{n-1}).
\end{equation}
where the symbol $\vee$ denotes concatenation of configurations.

This condition implies the existence of a unique
$p$-adic measure $\mu$ defined on $\Omega$ such that, for all $n$ and
$\sigma_n\in \Omega_{V_n}$, $\mu(\{\sigma|_{V_n}=\sigma_n\}) = \mu^{(n)}(\sigma_n)$.

\begin{defn} Measure $\mu$ defined by (\ref{e2.1a}), (\ref{e2.2}) is called a $p$-adic splitting (hard core) Gibbs measure, corresponding to the function $z$.
\end{defn}

The following statement describes conditions on the function $z$ that ensure compatibility of measures $\mu^{(n)}$.

\begin{pro}\label{p2.1} Probability measures $\mu^{(n)}, n = 1, 2, . . . $, in (\ref{e2.1a}) are compatible iff
for any $x \in V$ the following equation holds:
\begin{equation}\label{e2.3}
z'_x=\prod_{y\in S(x)}{\lambda+z'_y\over z'_y},
\end{equation}
here $z'_x={z_{0,x}\over z_{1,x}}$, $\lambda=\exp_p(J)$.
\end{pro}
\begin{proof} The proof consists in checking condition (\ref{e2.2}) for the measures (\ref{e2.1a}). It is analogous to that of Proposition 2.1 in \cite{SR}.
\end{proof}

Without loss of generality, we set hereafter $z_{1,x}=1$ and $z_x
= z'_x=z_{0,x}\in {\mathcal E}_p$. Then condition (\ref{e2.3})
reads
\begin{equation}\label{e3.1}
z_x=\prod_{y\in S(x)}{\lambda+z_y\over z_y}.
\end{equation}

\section{Conditions of solvability of equation (\ref{e3.1})}

In this section, we examine the conditions on the parameters $k\geq 1$, $p$ and $\lambda$ for the existence of solutions of equation (\ref{e3.1}). Notice that by Lemma \ref{l1} we have $\lambda=\exp_p(J)\in {\mathcal E}_p$.

\begin{thm}\label{t2} If $p$ does not divide $2^k-1$, then the equation (\ref{e3.1}) has no solution $z_x\in {\mathcal E}_p$, $x\in V$.
\end{thm}
\begin{proof} Let $z_x\in {\mathcal E}_p$, $x\in V$ be a solution then from (\ref{e3.1}) we get
 $$
|z_x|_p=\prod_{y\in S(x)}\left|{\lambda+z_y\over z_y}\right|_p=\prod_{y\in S(x)}|\lambda+z_y|_p=
$$
$$\prod_{y\in S(x)}|\lambda-1+z_y-1+2|_p=
\left\{\begin{array}{ll}
1, \ \ \ \ \ \ \mbox{if} \ \ p\ne 2,\\[3mm]
<p^{-k/(p-1)}, \ \ \mbox{if} \ \ p=2.\\
\end{array}\right.
$$

We shall use the following (see Lemma 4.6 of \cite{36})

\begin{lemma}\label{l2} If $a_i\in  {\mathcal E}_p$ for all $i=1,\dots,m$. Then
$$\prod_{i=1}^ma_i\in {\mathcal E}_p.$$
\end{lemma}

Assume $S(x)=\{x_1,\dots,x_k\}$ then from (\ref{e3.1}) we get
 $$
|z_x-1|_p=\left|\prod_{i=1}^k{\lambda+z_{x_i}\over z_{x_i}}-1\right|_p=$$
$$\left|\prod_{i=1}^k(\lambda+z_{x_i})-\prod_{i=1}^kz_{x_i}\right|_p=
$$
$$\left|\lambda\sum_{i=1}^k\prod_{j=1\atop j\ne i}^kz_{x_j}+\lambda^2\sum_{i=1}^k\sum_{j=1\atop j\ne i}^k\prod_{q=1\atop q\ne i,j}^kz_{x_q}+\dots+ \lambda^{k-1}\sum_{i=1}^kz_{x_i}+\lambda^k\right|_p=$$
$$\left|\sum_{i=1}^k\left(\lambda\prod_{j=1\atop j\ne i}^kz_{x_j}-1\right)+\sum_{i=1}^k\sum_{j=1\atop j\ne i}^k\left(\lambda^2\prod_{q=1\atop q\ne i,j}^kz_{x_q}-1\right)+\dots\right.$$
\begin{equation}\label{e2}
\left.+ \sum_{i=1}^k\left(\lambda^{k-1}z_{x_i}-1\right)+
\left(\lambda^k-1\right)+\left(2^k-1\right)\right|_p.
\end{equation}
Now using Lemma \ref{l2}, we get from (\ref{e2})

$$ \mbox{RHS\ \ of \ \ (\ref{e2})}=
\left\{\begin{array}{ll}
1, \ \ \ \ \ \ \mbox{if} \ \ p \nmid 2^k-1,\\[3mm]
< p^{-1/(p-1)}, \ \ \mbox{if} \ \ p\mid 2^k-1.\\
\end{array}\right.
$$
Thus the solution $z_x$ belongs to ${\mathcal E}_p$ iff $p$ divides $2^k-1$.
\end{proof}
Using this theorem, for given $k\geq 1$, one can find values of the prime number $p$ for which the equation (\ref{e3.1}) may have solutions (see the following table for small values of $k$).\\

\begin{center}
\begin{tabular}{|l|l|}
\hline
 $k$ & $p$ \\
 \hline
$1$ & $\emptyset$\\
\hline
$2$ & 3\\
\hline
$3$ & 7\\
\hline
$4$ & 3, 5\\
\hline
$5$& 31\\
\hline
$6$& 3, 7\\
\hline
$7$& 127\\
\hline
$8$& 3, 5, 17\\
\hline
$9$& 7, 73\\
\hline
$10$& 3, 11, 31\\
\hline
\end{tabular}
\end{center}
\bigskip

The following theorem gives a sufficient condition for the existence of a solution.

\begin{thm}\label{t3} If  $p$ divides $2^k-1$ and $p$ does not divide $k+2$ then the equation (\ref{e3.1}) has at least one solution $z_x\in {\mathcal E}_p$, $x\in V$.
\end{thm}
\begin{proof} We shall prove that under conditions of theorem the equation (\ref{e3.1}) has a constant (translational-invariant) solution $z_x=z$, $\forall x\in V$. In this case from (\ref{e3.1}) we get
\begin{equation}\label{e3}
z=\left({\lambda+z\over z}\right)^k.
\end{equation}
This equation can be written as $F(z)=0$ with
$F(z)=z^{k+1}-(\lambda+z)^k.$
Since $|\lambda|_p=1$ we have that the polynomial $F(z)$ has only $p$-adic integer coefficients. Hence we shall check other conditions of Hensel's lemma.
Take $a_0=1$ then we have
$$
F(1)=-(\lambda^k+k\lambda^{k-1}+\dots+k\lambda)=-\left((\lambda^k-1)+k(\lambda^{k-1}-1)+\dots+k(\lambda-1)+(2^k-1)\right).
$$
 Since $|\lambda-1|_p\leq {1\over p}$, using Lemma \ref{l2} we get $|\lambda^m-1|_p\leq {1\over p}$
for each $m=1,2,...$\\
This means that $p$ divides all $\lambda^m-1$ and also divides $2^k-1$.
Consequently $p$ divides $F(1)$, i.e. $F(1)\equiv 0\,(\operatorname{mod} p)$. Now let us check that if $p$ does not divide $k+2$, then $F'(1)\neq 0\,(\operatorname{mod} p)$.
We have
$$
F'(1)=k+1-k(\lambda+1)^{k-1}=k+1-k\left((\lambda-1)+2\right)^{k-1}=
$$
 $$k+1-k\left((\lambda-1)^{k-1}+(k-1)(\lambda-1)^{k-2}2+\dots+(k-1)(\lambda-1)2^{k-2}+2^{k-1}\right).$$
 Since $p$ divides $\lambda-1$ we must have $p\nmid (k+1-k2^{k-1})$. Using $p|2^k-1$ one can see that
 $p\nmid (k+1-k2^{k-1})$ is equivalent to $p\nmid (k+2)$. Thus conditions of Hensel's lemma are satisfied for
 $F(z)$ hence there exists a unique $p$-adic integer $a$ such that
 $F(a)=0$ and $a\equiv a_0\,(\operatorname{mod} p)$, i.e. $F(z)=0$ has a solution $z=a$. Since $a_0=1$ and $a\equiv 1\,(\operatorname{mod} p)$ we conclude $a\in {\mathcal E}_p$.
This proves the theorem.
\end{proof}
As a corollary of this Theorem, we have the following

\begin{thm}\label{t3a} If  $p$ divides $2^k-1$ and $p$ does not divide $k+2$ then for the $p$-adic HC model on Cayley tree of order $k\geq 1$ there exists at least one $p$-adic (splitting) Gibbs measure.
\end{thm}

\section{Uniqueness of translational-invariant measure}

In the previous section under conditions of Theorem \ref{t3} we have shown that equation (\ref{e3})
has at least one solution. Consequently by Proposition \ref{p2.1} there exists at least one translational invariant $p$-adic Gibbs measure. The following theorem asserts that such a measure is unique.

  \begin{thm}\label{t4} If  $p$ divides $2^k-1$ and $p$ does not divide $k+2$ then there exists a unique translational invariant $p$-adic Gibbs measure.
\end{thm}
\begin{proof} We shall prove that the equation (\ref{e3})
has a unique solution $z=a\in {\mathcal E}_p$. Assume that there are two such solutions $a$ and $b$, $a\ne b$.
Then we have $F(a)=F(b)=0$. Hence
$$F(a)-F(b)=(a-b)\left((a^k+a^{k-1}b+\dots+b^k)-\right.$$ $$\left.((a+\lambda)^{k-1}+(a+\lambda)^{k-2}(b+\lambda)+\dots+(b+\lambda)^{k-1})\right)=0.$$
Since $a\ne b$ we get
$$a^k+\dots+b^k=(a+\lambda)^{k-1}+\dots+(b+\lambda)^{k-1}$$ which can be written
$$(a^k-1)+\dots+(b^k-1)+k+1=[(a-1)+(\lambda-1)]^{k-1}+\dots+(k-1)[(a-1)+(\lambda-1)]2^{k-2}+$$ $$[(b-1)+(\lambda-1)]^{k-1}+\dots+(k-1)[(b-1)+(\lambda-1)]2^{k-2}+
k2^{k-1}.$$ Consequently,
$$2(a^k-1)+\dots+2(b^k-1)-2[(a-1)+(\lambda-1)]^{k-1}-\dots-2(k-1)[(a-1)+(\lambda-1)]2^{k-2}-$$ $$2[(b-1)+(\lambda-1)]^{k-1}-\dots-2(k-1)[(b-1)+(\lambda-1)]2^{k-2}=k(2^k-1)-(k+2).$$
This equality is not satisfied for any $a,b\in {\mathcal E}_p$ since the $p$-adic norm of its LHS is $\leq {1\over p}$ while the $p$-adic norm of the RHS is 1.
\end{proof}

\section{Periodic $p$-adic measures}

In this section, we shall consider periodic measures and use the
group structure of the Cayley tree. It is known (see \cite{yp1})
that there exists a one-to-one correspondence between the set of
vertices $V$ of a Cayley tree of order $k\geq 1$ and the group
$G_k$, free product of $k+1$ second-order cyclic groups with
generators $a_1, a_2, . . . , a_{k+1}$.

\begin{defn} Let ${\tilde G}$ be a normal subgroup of the group $G_k$. The set $z = \{z_x: x\in G_k\}$
 is said to be ${\tilde G}$-periodic if $z_{yx} =z_x$ for any $x\in G_k$ and $y\in {\tilde G}$.
 \end{defn}

\begin{defn} The ($p$-adic) Gibbs measure corresponding to a ${\tilde G}$-periodic set of quantities $z$ is said to be ${\tilde G}$-periodic.
\end{defn}

It is easy to see that a $G_k$-periodic measure is translational invariant.
Denote

$$G^{(2)} = \{x\in G_k: \, \mbox{the length of word} \, x \, \mbox{is even}\}.$$
This set is a normal subgroup of index two \cite{yp1,yp2}.

The following theorem characterizes the set of all periodic measures.

\begin{thm}\label{ty} Let $\tilde{G}$ be a normal subgroup of finite index in $G_k$.
Then each $\tilde{G}$- periodic $p$-adic Gibbs measure for HC model is either translation-invariant or $G^{(2)}$-
periodic.
\end{thm}
\begin{proof} Denote $f(z)={\lambda+z\over z}$. It is easy to check that $f(z)= f(t)$ if and only if $z = t$.
This property together with arguments similar to the ones given in the proof of Theorem 2 in \cite{MRS} lead to the statement.
\end{proof}

Let $\tilde{G}$ be a normal subgroup of finite index in $G_k$.
Let us state condition on $\tilde{G}$ for each $\tilde{G}$-periodic $p$-adic Gibbs measure to be translation invariant.\\
Set $I(\tilde{G}) = \tilde{G}\cap\{a_1, ..., a_{k+1}\}$, where the $a_i$ are the generators of $G_k$.

\begin{thm}\label{ty1} If $I(\tilde{G})\ne\emptyset$ then each $\tilde{G}$- periodic $p$-adic Gibbs measure is translational-invariant.
\end{thm}
\begin{proof} Similar to proof of Theorem 3 in \cite{MRS}.\end{proof}

By Theorems \ref{ty} and \ref{ty1}, the description of a $\tilde{G}$-periodic $p$-adic Gibbs measure for $I(\tilde{G})\ne \emptyset$
reduces to finding that of fixed points of the map $(f(z))^k$ (these fixed points correspond to translational invariant $p$-adic Gibbs measures). \\
For $I(\tilde{G}) =\emptyset$, it reduces to the solutions of
system (\ref{ep}) below. This system describes periodic measures
with period two, more precisely, $G^{(2)}$-periodic $p$-adic
measures. They correspond to functions
$$z_x =\left\{\begin{array}{ll}
z_1, \ \ \mbox{if} \ \ x\in G^{(2)},\\
z_2, \ \ \mbox{if} \ \ x\in G_k \setminus G^{(2)}.
\end{array}\right.
$$
In this case, we have from (\ref{e3.1}):
\begin{equation}\label{ep}
z_1 =\left({z_2 +\lambda\over z_2}\right)^k, \ \ z_2 =\left({z_1 +\lambda\over z_1}\right)^k.
\end{equation}
Namely,  $z_1$ and $z_2$ satisfy
\begin{equation}\label{ef}
z = g(g(z)), \ \ \mbox{where}\ \ g(z) = ((z+\lambda)/z)^k.
\end{equation}
Note that to get periodic (non translational invariant) measure we must find solutions $z_1,z_2\in {\mathcal E}_p$
of (\ref{ep}) with $z_1\ne z_2$.
Obviously, such solutions are roots of the equation
\begin{equation}\label{ee} {g(g(z)) - z\over g(z) - z}= 0,
\end{equation}
which is equivalent to the equation
\begin{equation}\label{zz}
{L(z)\over M(z)}=0, \ \ \mbox{with}\ \ L(z)=\left(\lambda z^k+(\lambda+z)^k\right)^k-z(\lambda+z)^{k^2}; \ \ M(z)=(\lambda+z)^k-z^{k+1}.
\end{equation}
We have
$$L(z)=\left((\lambda+z)z^k+M(z)\right)^k-z\left(z^{k+1}+M(z)\right)^k=(\lambda+z)^kz^{k^2}+  $$
$$\sum_{i=1}^k{k\choose i}M^i(z)((\lambda+z)z^k)^{k-i}-z^{k^2+k+1}-
z\sum_{j=1}^k{k\choose j}M^j(z)(z^{k+1})^{k-j}=M(z)U(z),$$
where
$$
U(z)=(1-k)z^{k^2}+k\left((\lambda+z)z^k\right)^{k-1}+$$ $$\sum_{i=2}^k{k\choose i}M^{i-1}(z)z^{k(k-i)}\left((\lambda+z)^{k-i}-z^{k-i+1}\right).$$

Hence in order to get $G^{(2)}$-periodic (not translation invariant)
solutions of (\ref{ep}) we must find solutions of equation $U(z)=0$.
Conditions for existence of such solutions are given in the following theorem.
 \begin{thm}\label{t4a}  Let $p\ne 3,5$. The equation $U(z)=0$ has a solution $z\in {\mathcal E}_p$ if and only if $p$ divides $2^k-1$ and $p$ divides $k-2$.
\end{thm}
\begin{proof} {\it Necessity:}  from above, it follows that if $p$ divides $2^k-1$ then $|M(z)|_p\leq {1\over p}$. The function $U(z)$ can be written as
$$
U(z)=(1-k)(z^{k^2}-1)+k\left((\lambda-1+z-1)z^k\right)^{k-1}+k(2^{k-1}-1)+1+$$ $$\sum_{i=2}^k{k\choose i}M^{i-1}(z)z^{k(k-i)}\left((\lambda+z)^{k-i}-z^{k-i+1}\right).$$
 Since $p$ divides $M(z)$, it must divide $(k(2^{k-1}-1)+1)$. By using $p\mid 2^k-1$ one can see that $p\mid (k(2^{k-1}-1)+1)$ is equivalent to $p\mid (k-2)$.

 {\it Sufficiency:} Since $|\lambda|_p=1$, the polynomial $U(z)$ has only $p$-adic integer coefficients. Hence we shall check other conditions of Hensel's lemma.
Take $a_0=1$ then it is easy to see that $U(1)\equiv 0\,(\operatorname{mod} p)$. Now we shall check that $U'(1)\ne 0\,(\operatorname{mod} p)$. We have
$$U'(1)=(1-k)k^2+k(k-1)(k\lambda+k+1)(\lambda+1)^{k-1}+$$ $${k(k-1)\over 2}\left(k(\lambda+1)^{k-1}-(k+1)\right)\left((\lambda+1)^{k-2}-1\right)+pN,$$
where $N\in\mathbb N$.
Now using hypothesis of theorem, we get
$$U'(1)=15+pN_1, \ \ N_1\in \mathbb N.$$
Hence if $p\ne 3,5$ all conditions of Hensel's lemma are satisfied.
This completes the proof.
\end{proof}

For given $k\geq 1$, one can easily find values of prime number $p$ for which the equation $U(z)=0$ has a solution (see the following table for small values of k)\\

\begin{center}
\begin{tabular}{|l|l|}
\hline
$k$   &  $p$ \\
 \hline
$1$ & $\emptyset$\\
\hline
$2$ & 3\\
\hline
$3$ & $\emptyset$\\
\hline
$4$ & $\emptyset$\\
\hline
$5$& $\emptyset$\\
\hline
$6$& $\emptyset$\\
\hline
$7$& $\emptyset$\\
\hline
$8$& 3\\
\hline
$9$& 7\\
\hline
$10$& $\emptyset$\\
\hline
\end{tabular}
\end{center}
\bigskip

Now we are going to give all $G^{(2)}$-periodic solutions for $k=2$. In this case the equation (\ref{ee}) has the following form:
\begin{equation}\label{ee2}
z^2-(\lambda^2-2\lambda)z+\lambda^2=0.
\end{equation}
The solutions of this quadratic equation are
\begin{equation}\label{eq}
z_{1,2}={\lambda\over 2}\left(\lambda-2\pm\sqrt{\lambda(\lambda-4)}\right).
\end{equation}
We must check the existence of $\sqrt{\lambda(\lambda-4)}$ and additionally that $z_{1,2}\in {\mathcal E}_p$.
For $k=2$,  following Theorem \ref{t2}, only the case $p=3$ has to be considered. Since $\lambda\in {\mathcal E}_3$ we have its following representation
$$\lambda=1+\lambda_1\cdot 3+\lambda_2\cdot 3^2+\lambda_3\cdot 3^3+\cdots.$$
It is easy to see that $\lambda$ satisfies hypothesis of Theorem \ref{tx2}. Hence $\sqrt{\lambda}$ exists in $\Q_3$.
So we must check the existence of $\sqrt{\lambda-4}$ in $\Q_3$. We have
$$-3=3\cdot {2\over 1-3}=2\cdot 3+2\cdot 3^2+2\cdot 3^2+\cdots.$$ Consequently
$$\lambda-4=-3+\lambda_1\cdot 3+\lambda_2\cdot 3^2+\lambda_3\cdot 3^3+\cdots
=3\left((\lambda_1+2)+(\lambda_2+2)\cdot 3+(\lambda_3+2)\cdot 3^2+\cdots\right).$$
From this equality and Theorem \ref{tx2}, it follows that $\lambda_1=1$ ensures the existence of $\sqrt{\lambda-4}$.
Then we have
$$\lambda-4=3^2\left(\lambda_2+\lambda_3\cdot 3+\lambda_4\cdot 3^2+\lambda_5\cdot 3^3+\cdots\right)$$
which implies  that $\lambda_2$ must be
a quadratic residue modulo $3$ by refering again to Theorem \ref{tx2}. This leads to $\lambda_2=1$ only, therefore
$$\lambda-13=3^3\left(\lambda_3+\lambda_4\cdot 3+\lambda_5\cdot 3^2+\cdots\right), \ \ 0\leq \lambda_i\leq 2, \, i=3,4,5,\cdots.$$
Now we check that $z_{1,2}\in {\mathcal E}_3$. We have
\begin{equation}\label{eq1}
|z_{1,2}|_3=
\left|\lambda-2\pm\sqrt{\lambda(\lambda-4)}\right|_3=\left|(\lambda-1)-1\pm\sqrt{\lambda((\lambda-1)-3)}\right|_3=1.
\end{equation}
\begin{equation}\label{eq2}
|z_{1,2}-1|_3=\left|(\lambda-1)^2-3\pm\lambda\sqrt{\lambda((\lambda-1)-3)}\right|_3<3^{-1\over 2}.
\end{equation}
Hence we have proven the following
 \begin{thm}\label{t5} If  $k=2$, $p=3$ and $\lambda\in \{x\in \Q_3: |x-13|_3\leq {1\over 27}\}$ then there exist precisely two $G^{(2)}$-periodic $p$-adic Gibbs measures.
\end{thm}
\begin{rk}
In classical (real) models of statistical mechanics, a phase transition is said to occur whenever varying a parameter leads to a change in the number of Gibbs states. For example, on a Cayley tree, Ising and Potts models exhibit a phase transition at some critical temperature $T_c$. Similar phenomena also occurs for real HC model at some $\lambda_c$. This is not the case for $p$-adic models since the field of $p$-adic numbers $\Q_p$ is not ordered. However in the case $k=2, p=3$ the sphere $\{x\in \Q_p: |x-13|_p={1\over 27}\}$ can be considered as a critical ``curve".

Note that in $p$-adic case the geometry of balls and spheres are more complicated than in real case \cite{22,23,24,29,40,41,48}.
\end{rk}


\section{Boundedness of $p$-adic Gibbs measures}

Now we are interested to find out whether a $p$-adic Gibbs measure is bounded.

For a set $A$ we denote by $|A|$ its number of elements
and recall that $\Omega_n$ is the set of all admissible configurations $\sigma_n:V_n\to \{0,1\}$.
We need the following

\begin{lemma}\label{l3} The number of admissible configurations is given by
 $$|\Omega_n|=2^{(k+1){k^n-1\over k-1}}+1.$$
\end{lemma}
\begin{proof} We first compute the number of non admissible configurations.
It is known that if a connected subset $M$ of a tree contains $m$ vertices then it contains
$m-1$ edges. Thus $V_n$ contains $|V_n|-1$ edges. Note that a configuration $\sigma_n$ is
non admissible if there exists at least one edge $\langle x,y\rangle$ such that $\sigma(x)=\sigma(y)=1$.
Therefore, the number of non admissible configurations on $V_n$ is equal to
$$\sum_{m=1}^{|V_n|-1}{|V_n|-1\choose m}=2^{|V_n|-1}-1.$$
Consequently
$$|\Omega_n|=2^{|V_n|}-\left(2^{|V_n|-1}-1\right)=2^{|V_n|-1}+1.$$
This, together with the following formula
$$|V_n|=1+(k+1){k^n-1\over k-1}$$
completes the proof
\end{proof}

\begin{thm} Assume $p\mid 2^k-1$, then the $p$-adic Gibbs measure $\mu$ corresponding to the $p$-adic HC-model on the Cayley tree of order $k\geq 1$ is bounded if and only if $p\ne 3$.
\end{thm}
\begin{proof} It suffices to show that the values of $\mu$ on cylindrical subsets are
bounded.

Denote
$$\tilde{H}(\sigma_n)=J\sum_{x\in V_n}\sigma(x)+\sum_{x\in W_n}\log_pz_{\sigma(x),x}.$$

Let us estimate $\left|\mu^{(n)}(\sigma_n)\right|_p$:
$$
\left|\mu^{(n)}(\sigma_n)\right|_p=\left|{\exp_p\left(\tilde{H}(\sigma_n)\right)\over
\sum_{\varphi_n\in \Omega_n}\exp_p\left(\tilde{H}(\varphi_n)\right)}\right|_p=$$
\begin{equation}\label{eb}
{1\over\left|
\sum_{\varphi_n\in \Omega_n}\left[\exp_p\left(\tilde{H}(\varphi_n)\right)-1\right]+|\Omega_n|\right|_p}.
\end{equation}
Using Lemma \ref{l3} we get
$$|\Omega_n|=2^{k{\bf K}+1}+1=2\left[(2^k-1)^{{\bf K}}+{\bf K}(2^k-1)^{{\bf K}-1}+\dots+(2^k-1)\right]+3,$$
where ${\bf K}=2+2k+\cdots+2k^{n-2}+k^{n-1}$.
Consequently,
$$\left| |\Omega_n|\right|_p=\left\{\begin{array}{ll}
\leq {1\over p} \ \ \mbox{if} \ \ p=3\\[2mm]
1\ \ \mbox{if} \ \ p\ne 3.\\
\end{array}\right.
$$ Now from (\ref{eb}) we get
$$\left|\mu^{(n)}(\sigma_n)\right|_p=\left\{\begin{array}{ll}
\geq p \ \ \mbox{if} \ \ p=3\\[2mm]
1\ \ \mbox{if} \ \ p\ne 3.\\
\end{array}\right.$$
Thus boundedness is proved for $p\ne 3$.

Now we shall prove that $\mu$ is not bounded if $p=3$.
Put
$$p^{x,y}_{ij} =\left\{\begin{array}{ll}
{\exp_p\left(J(i+j)+z_{i,x} +z_{j,y}\right)\over\sum_{u,v\in\{0,1\}\atop u+v\ne 2} \exp_p\left(J(u+v)+z_{u,x} +z_{v,y}\right)} \ \ \mbox{if} \ \ i,j=0,1; \, i+j\ne 2\\[6mm]
0 \ \ \ \ \ \ \mbox{if} \ \ i=j=1.
\end{array}\right.$$

In order to show that the measure $\mu$ is not bounded at $p=3$, it is enough to show that
its marginal measure is not bounded. Let $\pi=\{. . . , x_{-1}, x_0, x_1, . . . \}$ be an arbitrary
infinite path in $\Gamma^k$. From (\ref{e2.1a}) we can see that a marginal measure $\mu_\pi$ on admissible configurations on $\{0,1\}^\pi$ has the form
\begin{equation}\label{m}
\mu_\pi(\omega_n)=p^{x_{-n},x_{-n+1}}_{\omega(x_{-n})}\prod^{n-1}_{m=-n}p^{x_m,x_{m+1}}_{\omega(x_m)\omega(x_{m+1})}
\end{equation}
Here $\omega_n : \{x_{-n}, ..., x_0, ..., x_n\}\to \{0,1\}$ is a configuration on
$\{x_{-n}, . . . , x_0, . . . , x_n\}$ and $p^{xy}_i$ is a coordinate of the invariant vector of the
matrix $\left(p^{x,y}_{ij}\right)_{i,j=0,1}$.

We have
\begin{equation}\label{k}
\left|p^{x,y}_{ij}\right|_3 = {1\over \left|\sum_{u,v\in\{0,1\}\atop u+v\ne 2} \left[\exp_3\left(J(u+v)+z_{u,x} +z_{v,y}\right)-1\right]+3\right|_3}>3
\end{equation}
for all $i, j$. From (\ref{m}) and (\ref{k}) we find that $\mu_n$ is not bounded. \\
The theorem is proven.
\end{proof}

\section{Concluding remarks} To conclude, we will give a brief description
of the differences of behavior between  classical (real) models and $p$-adic models on Cayley trees.

{\bf Hard core models.} {\it Real case:} In this model (see \cite{SR}), for all $\lambda > 0$
and $k\geq 1$, there exists a unique translational invariant splitting Gibbs measure $\mu_0$.
Let $\lambda_{\rm c} = {1\over(k-1)}\left({k\over k-1}\right)^k$, then: \\
(i) for $\lambda\leq \lambda_{\rm c}$, the Gibbs measure is unique (and coincides with the above measure $\mu_0$),\\
(ii) for $\lambda > \lambda_{\rm c}$, in addition to $\mu_0$, there exist two distinct extreme periodic measures, $\mu^+$ and $\mu^-$. In addition, there are a continuum set of distinct, extreme, non-translational-invariant, Gibbs measures.\\
For $\lambda > {1\over(\sqrt{k}-1)}\left({\sqrt{k}\over \sqrt{k}-1}\right)^k$, the measure $\mu_0$ is not extreme.

{\it $p$-adic case:} In this paper we have shown that the $p$-adic HC model is completely different from real HC model. For a fixed $k$, the $p$-adic HC model may have a splitting Gibbs measure only if $p$ divides $2^k-1$. Moreover if $p$ divides $2^k-1$ and does not divide $k+2$ then there exists unique translation invariant $p$-adic Gibbs measure. The HC model admits only translation invariant and periodic with period two (chess-board) Gibbs measures. For $p\geq 7$, a periodic $p$-adic Gibbs measure exists iff $p$ divides both $2^k-1$ and $k-2$. For $k=2$, a $p$-adic splitting Gibbs measure exists if and only if $p=3$;
in this case we have shown that if $\lambda$ belongs to a $p$-adic ball of radius $1/27$ then there are precisely two periodic (non translation invariant) $p$-adic Gibbs measures. Finally, we have proven that a $p$-adic Gibbs measure is bounded if and only if $p\ne 3$.\\

{\bf Potts model.}
{\it  Real case:}  The ferromagnetic $q$ states Potts model for any $q\geq 2$ exhibits possibly $q+1$ distinct translation invariant Gibbs measures.
Namely, there exist two critical temperatures $0<T'_{\rm c}<T_{\rm c}$ such that:\\
(i) for $T\in (T'_{\rm c},T_c]$ there are $q + 1$ extreme Gibbs measures, one of them
is called unordered;\\
(ii) for $T\le T'_{\rm c}$, $q$ extreme Gibbs measures coexist:
there is the unordered which is not extreme;\\ (iii) for $T
>T_{\rm c}$ there is one Gibbs measure, \cite{G, M}.

{\it $p$-adic case:}
The model exhibits a phase transition whenever  $k = 2, q\in p\mathbb N$
and $p\geq 3$ (resp. $q \in  2^2\mathbb N, p = 2$)  \cite{36}.
Whenever $k\geq 3$ a phase transition may occur only at $q \in p\mathbb N$ if $p\geq 3$ and $q\in 2^2\mathbb N$ if $p = 2$.
Moreover for the $p$-adic Ising model ($q=2$)  there is no phase transition. This is one interesting difference between real and $p$-adic Ising model on trees.

{\bf $\lambda$-model.} {\it Real case:} A nearest-neighbor $\lambda$-model with two spin values on Cayley tree is considered in \cite{Ro}. There, it was proven that this model has similar properties like Ising model.

{\it $p$-adic case:} (see \cite{19}) For $p$-adic  non homogeneous $\lambda$-model there is no phase transition and as well
as being unique, the $p$-adic Gibbs measure is bounded if and only if $p\geq 3$. If $p=2$, a condition asserting  the non existence of a phase transition was given.

This result shows that, in $p$-adic case, even non homogeneous interactions do not lead to the ocurence of a phase transition.

\smallskip

From the above given results it follows that the set of $p$-adic Gibbs measures is sparse with respect to  the set of real Gibbs measures. The main reasons for this could be explained by the following:

(i) The set of values of real norm $|x|$ is continuous $[0,+\infty)$, but the set of values of $p$-adic norm is discrete $\{p^m: m\in \Z\}$.

(ii) Real function $e^x$ is defined for any $x\in R$ but $p$-adic function $\exp_p(x)$ is defined only for $x\in \Q_p$ with $|x|_p\leq {1\over p}$.

(iii) The set of values of real function $e^x$ and its norm $|e^x|$ is continuous $(0,+\infty)$, but the set of values of $p$-adic function $\exp_p(x)$ is $\{x: |x-1|_p\leq {1\over p}\}$ and the set of values of its norm $|\exp_p(x)|_p$ contains unique point 1, i.e., $ |\exp_p(x)|_p=1$ for all $x$ with $|x|_p\leq {1\over p}$.

Nevertheless, we believe that $p$-adic Gibbs measures might have interesting applications.

\section*{ Acknowledgements}

UAR thanks the universit\'e du Sud Toulon Var for supporting his visit to Toulon and the Centre de Physique Th\'eorique de Marseille for kind hospitality. He also acknowledge IMU/CDE-program for a travel support and the TWAS  Research Grant: 09-009 RG/Maths/As-I; UNESCO FR: 3240230333.

\end{document}